\documentclass[11pt, tightenlines, twoside, onecolumn, nofloats, nobibnotes, nofootinbib, superscriptaddress, noshowpacs, centertags]{revtex4-2}
\usepackage{amsthm}
\newtheorem*{theorem*}{Theorem}
\usepackage{ljm}
\begin{document}

\titlerunning{On the extension of singular linear} 
\authorrunning{V.A. Glazatov, V.Zh. Sakbaev} 

\title{On the extension of singular linear infinite-dimensional
Hamiltonian flows\\
}

\author{\firstname{V.~A.}~\surname{Glazatov}}
\email[E-mail: ]{glazv96@yandex.ru}
\affiliation{Keldysh Institute of Applied Mathematics (Russian Academy of Sciences), 4 Miysskaya sq., Moscow 125047, Russia}
\affiliation{Moscow Institute of Physics and Technology, 9 Institutskiy per., Dolgoprudny, Moscow Region, 141701, Russia}

\author{\firstname{V.~Zh.}~\surname{Sakbaev}}
\email[E-mail: ]{fumi2003@mail.ru}
\affiliation{Keldysh Institute of Applied Mathematics (Russian Academy of Sciences), 4 Miysskaya sq., Moscow 125047, Russia}


\received{May 1, 2024} 

\begin{abstract} 
We  study Hamiltonian flows in a real separable Hilbert space endowed with a symplectic structure. Measures on the Hilbert space that are invariant with respect to the flows of completely integrable Hamiltonian systems are investigated. These construction gives the opportunity to describe Hamiltonian flows in the phase space by means of unitary groups in the space of functions that are quadratically integrable by the invariant measure. Invariant measures are applied to the study of model linear Hamiltonian systems that admit features of the type of unlimited increase in kinetic energy over a finite time. Due to this approach solutions of  Hamilton equations that admit singularities can be described by means of the phase flow in the extended phase space and by the corresponding Koopman representation of the unitary group.
\end{abstract}

\subclass{Primary 28C20; Secondary 28A35} 

\keywords{symplectomorphism, translation-invariant measure, A. Weil's theorem, Hamiltonian flow, Koopman representation} 

\maketitle

\section{Introduction}
In a number of problems of mathematical physics arise Hamiltonian systems,
whose phase space is an infinite-dimensional separable Hilbert space. In particular, the space
endowed with a translationally invariant symplectic form \cite{Bourgain, Zhidkov, Cher-Mars}.
In some situations, the phase trajectories of an infinite-dimensional Hamiltonian system allow it to go to infinity in a finite time.
The article examines a model example of the phenomenon of going to infinity. This example is described by a linear system of Hamilton equations.

This paper studies quadratic Hamiltonian functions on a real separable Hilbert space.
This space is endowed with a symplectic structure. Also, on this space there is a non-negative finitely additive measure,
invariant under the group of symplectomorphisms \cite{Bourgain, S23, BS-direct}.
We consider a quadratic form of the Hamiltonian, which is not majorized by the quadratic form of the Hilbert norm from above,
neither from below. It is shown that the trajectories of such Hamiltonian systems allow them to go to infinity in the phase space in a finite time.

To describe the phase flow of Hamiltonian systems,
whose trajectories go to infinity in a finite time, the extended phase space is introduced.
The phase space extension is a locally convex space. The original Hilbert space is tightly embedded in this space.
The Hamiltonian function, densely defined on the original Hilbert space, can be extended to the extended phase space,
trajectories of a Hamiltonian system, symplectic form and invariant measure.
A new class of invariant measures on the extended phase space is also presented.
In this case, the extended Hamilton function is densely defined on the extended phase space.
Besides, the extended symplectic form is not a restricted bilinear form.
But the extended phase flow preserves the extended symplectic form and the measure on the extended phase space.

The presence of the measure invariant with respect to the Hamiltonian flow on the extended phase space
allows us to obtain the Koopman unitary representation of the flow in the Hilbert space of functions,
quadratically integrable with respect to the invariant measure. The Koopman unitary group is not continuous in the strong operator topology.
An analysis of the spectral properties of operators of the Koopman group is given. This allows us to find invariant subspaces,
 the restriction to which the unitary group has the property of strong continuity.

An important role in the implementation of the proposed research program for infinite-dimensional
Hamiltonian systems are played by the measure on the phase space that is invariant under the group of Hamiltonian transformations.

The study of measures on topological vector spaces that are invariant under transformation groups,
meets Weyl's theorem. Therefore, we need to analyze measures that do not have some of the properties of the Lebesgue measure.

Finitely additive measures, including analogues of the Lebesgue measure on infinite-dimensional locally convex
spaces, have applications to the study of quantization of infinite-dimensional Hamiltonian systems (in particular, secondary quantization).
Also, such measures are needed for problems of the statistical mechanics, the theory of quantum changes, for the study of random unitary groups and the dynamics of open quantum systems \cite{Srinivaz, KS-18, SSh, OSS-16, S23}.
One of the important properties of the Lebesgue measure on a finite-dimensional Euclidean space as an Abelian topological group with the operation of addition of elements
is not only its invariance under the action of an arbitrary element of the group (i.e., a shift by an arbitrary vector),
but also with respect to shifts along the trajectories of divergence-free vector fields, in particular, with respect to Hamiltonian transformations.

The group of translations by space vectors is a subgroup of the group of Hamiltonian flows,
 generated by Hamiltonian functions linear in coordinates and momenta. Therefore, the construction of translation invariant measures on
locally convex spaces is an important step in the study of the problem posed
\cite{Baker, Pantsulaya,
Zavadskii, SSh, s16}.
In the works \cite{GS22, S23, BS-direct} the continuation of the measure from the work \cite{s16} to a wider ring of subsets was studied,
invariant under flows generated by certain Hamiltonian fields.
We will call such extensions symplectic measures.

The class of quadratic Hamiltonians (hyperbolic oscillators) on the Hilbert phase space is studied.
For this class, solutions to the linear system of Hamilton equations allow the phenomenon of an unlimited increase in kinetic energy
in a finite time. It is shown that the dynamics of such Hamiltonian systems admits a natural continuation from the Hilbert phase
 space onto the locally convex phase space containing it. Namely, a continuation with
implectic form from the Hilbert space to the topological vector space of number sequences.
For such a continuation, the phase flow allows a single coordinate-wise continuation into the extended phase space.
Invariant measures on extended phase space
are constructed as products of finitely additive translation invariant measures on
countable collection of finite-dimensional Euclidean spaces. This is necessary to obtain the Koopman representation of the flow.

This work continues the studies of the blow-up phenomenon in linear systems begun in \cite{GS22, S23}.
A new design for expanding the phase space and continuing phase trajectories is proposed.
The introduction of invariant measures on the extended phase space made it possible to describe the Koopman representation
phase flow in terms of spectral properties of unitary transformations.

Section 2 gives a description of a homogeneous symplectic structure on a separable Hilbert space.

In section 3
a finitely additive measure was constructed on a real seperable Hilbert space,
equipped with a standard symplectic structure. The resulting measure is invariant under Hamiltonian flows,
preserving two-dimensional symplectic subspaces.

Section 4 considers applications of the invariant measure to Hamiltonian systems.
Classes of Hamiltonian systems, including linear and quadratic, have been studied.

Section 5 defines the procedure for expanding the phase space and the procedure for continuing the trajectories of a Hamiltonian system,
leaving the original phase space in a finite time into the expanded space.

In Section 6, the Koopman representation of the Hamiltonian flow in the extended phase space is obtained. For this purpose it is used
unitary group in the space of functions that are quadratically integrable with respect to a flow-invariant measure.

\section{Symplectic structure}
A symplectic structure on a real separable Hilbert space $E$ is a non-degenerate closed differential 2-form on the space $E$. If the symplectic structure is on
Hilbert space $E$
is invariant under translations, then it is given by a non-degenerate skew-symmetric bilinear form $\omega$ on the space $E$ (in this case, the Hilbert space $E$ is identified with its conjugate).
A symplectic structure $\omega$ on a real separable Hilbert space $E$ is said to be natural if in the space $E$ there exists an orthonormal basis $\{ g_k\}\equiv \mathcal G$ such that $\omega (g_{2k-1 },g_j)=\delta _{j,2k},\ k,j\in \mathbb N$, $\delta _{j,i}$ is a Kronecker symbol (see \cite{KS-18, SSh}).

The natural symplectic structure $\omega$ determines the decomposition of the space $E$ into the direct sum of two subspaces $Q\oplus P$. The bases of this subspaces are respectively orthonormal systems $e_j=g_{2j-1},\, j\in \mathbb N $ and $f_k=g_{2k},\ k\in \mathbb N$. Then
\begin{equation}\label{simp}
\omega (e_j,e_i)=0,\quad \omega (f_i,f_j)=0\quad \forall \quad i,j\in {\mathbb N};\qquad \omega (e_j,f_k)=\delta _{jk},\, j,k\in \mathbb N
\end{equation}
(see \cite{KS-18}). In this case, the basis $\{ g_i,\ i\in \mathbb N \}=\{ e_j,f_k;\ j,\, k\in {\mathbb N}\}$ is called the symplectic basis of the space $E$, corresponding to the symplectic form $\omega$.

The linear operator ${\bf J}$ associated with the bilinear form $\omega $ is a non-degenerate skew-symmetric operator. Its action on the vectors of the symplectic basis
is given by the equalities $${\bf J}(e_j)=-f_j,\quad {\bf J}(f_k)=e_k,\quad j\in {\bf N},\ k\in {\bf N}. $$
In this case, $Q$ and $P$ are called a coordinate space and a momentum space, respectively, and it is assumed that $P$ is the conjugate of $Q$ (see \cite{Khrennikov, KS-18, SSh}).

A Hamiltonian system is a triple $(E,{\bf J},h),$ where $(E,{\bf J})$ is a Hilbert space with a symplectic structure, $h:E_2\to \mathbb R$ is define and continuously differentiable according to Gateaux on the vector subspace $E_2$ of the space $E$ a function called the Hamilton function.
A densely defined vector field ${\bf v}:\ E_2\to E$ is called Hamiltonian if there exists a Hamiltonian function $h:\ E_1\to \mathbb R$, $E_2\subset E_1\subset E$ such that
${\bf v}(z)={\bf J}Dh(z),\, z\in E_2$. Here the function $h$ is differentiable on the subspace $E_2\subset E_1$ densely embedded in the space $E$. $Dh$ is the differential of the function $h$. $\bf J$ is the linear operator associated with the bilinear form $\omega $ in the Hilbert space $E$.

An one-parameter group ${\Phi }_t,\, t\in {\mathbb R},$ of continuously differentiable transformations of the space $E_2$ is called a smooth Hamiltonian flow in the space $E_2$ generated by the Hamiltonian vector field ${\bf v}:\ E_2\to E$, if the equality $\frac{d }{dt}{\bf \Phi }_t(q,p)={\bf v}({\bf \Phi }_t(q,p))$ is satisfied $ ,\, (q,p)\in E_2$. If a Hamiltonian flow in the space $E_2$ admits an unique continuity extension from the space $E_2$ to the space $E$, then such a continuation of the flow is called a generalized Hamiltonian flow in the space $E$ generated by the Hamiltonian vector field $\bf v$ (Hamiltonian $ h$). Such an extension of a smooth Hamiltonian flow to a generalized flow exists if the smooth flow does not increase the norm of vectors in the space $E$.

\section{Measures invariant under symplectomorphisms}
Let us pose the problem of describing measures invariant under a certain group of Hamiltonian transformations on a real separable Hilbert space $E$, equipped with a translation invariant symplectic form $\omega$.
Let $E=Q\oplus P$ and ${\mathcal E}={\mathcal F}\bigcup {\mathcal G}$ be a symplectic basis of the form $\omega $ (see (\ref{simp})) .

\begin{definition}\label{def1}
The set $\Pi \subset E$ is called an absolutely measurable symplectic beam in the space
$E$ if there is such a symplectic
basis $\{ f_j,g_k,\ j\in {\mathbb N},\, k\in {\mathbb N}\}$ that is the set
$\Pi$ is given by the equality
\begin{equation}\label{Pi}
\Pi =\{ z\in E:\ ((z,f_i),(z,g_i))\in B_i,\, i\in {\mathbb N}\} ,
\end{equation}
where $B_i$ are Lebesgue measurable sets of the plane
${\mathbb R}^2$, satisfying the conditions $\sum\limits_{j=1}^{\infty }\max \{\ln (\lambda _2(B_j)),0\}<+\infty $ (here $\lambda _n$ is the Lebesgue measure in ${\mathbb R}^n$, $n\in \mathbb N$).
\end{definition}

We fix some symplectic basis
${\mathcal E}={\mathcal F}
\bigcup
{\mathcal G}$. Let ${\mathcal K}_{{\mathcal F},{\mathcal G}}(E)\equiv {\mathcal K}_{\mathcal E}(E)$ be the set of absolutely measurable symplectic bars having a form
(\ref{Pi}) in the selected basis
${\mathcal F}\bigcup {\mathcal G}$.

Let $\lambda _{{\mathcal F}, {\mathcal G}}:\ {\mathcal K}_{{\mathcal F}, {\mathcal G}}(E)\to [0,+\infty )$ is the set function defined by the equality
$$
\lambda _{{\mathcal F}, {\mathcal G}}(\Pi )=\prod\limits_{j=1}^{\infty }\lambda _2(B_j)=\exp(\sum\limits_{ j=1}^{\infty }\ln (\lambda _2(B_j)))
$$
subject to $\Pi \neq \oslash$; and $\lambda _{{\mathcal F},{\mathcal G}}(\Pi )=0$ in the case of $\Pi =\oslash $.
It is easy to see that if
$A,B\in {\mathcal K}_{\mathcal F},{\mathcal G}(E)$ in some ONB ${\mathcal F}\bigcup {\mathcal G}$, then $A \bigcap B\in {\mathcal K}_{{\mathcal F},{\mathcal F}}(G)$. In addition, the class of sets ${\mathcal K}_{{\mathcal F},{\mathcal G}}(E)$ and the set function $\lambda _{{\mathcal F}, {\mathcal G}}:\ {\mathcal K}_{{\mathcal F},{\mathcal G}}(E)\to [0,+\infty )$ are invariant under translation by any vector in the space $E$.
The set $\Pi \in {\mathcal K}_{{\mathcal F},{\mathcal G}}(E)$ of the form (\ref{Pi}) is denoted by the symbol $\times _{j=1}^{ \infty }B_j$.

Let $r_{{\mathcal F},{\mathcal G}}$ be the ring generated by the system of sets
${\mathcal K}_{{\mathcal F},{\mathcal G}}.$

\begin{lemma}\label{l3.1} \cite{GS22}
{\it Class $\Lambda$ of sets of the form $A=\Pi \backslash (\bigcup\limits^n_{i=1}\Pi _i),$ where $n\in {\mathbb N}_0$, $\ Pi ,\Pi _1,...,\Pi _n\in {\mathcal K}_{{\mathcal E},{\mathcal F}}$, is a semiring.}
\end{lemma}
\begin{theorem}\label{th3.6.} \cite{GS22} {\it The set function $\lambda :\ {\mathcal K}_{{\mathcal F},{\mathcal G}}(E) \to [0,+\infty )$ is additive. The additive function of the set $\lambda :\ {\mathcal K}_{{\mathcal E},{\mathcal F}}(E)\to [0,+\infty )$ admits a unique additive extension to the ring $r_{{ \mathcal F},{\mathcal G}}$ generated by the semiring $\Lambda $.}
\end{theorem}

Completion of a finitely additive measure
$\lambda :\ r_{{\mathcal F},{\mathcal G}} \to [0,+\infty )$ is the complete measure of $\lambda _{{\mathcal E},{\mathcal F}} :\ {\mathcal R}_{{\mathcal E},{\mathcal F}}\to [0,+\infty )$. The ring $r$ defines the ring ${\mathcal R}_{{\mathcal E},{\mathcal F}}$ as follows. The internal ${\underline \lambda }$ and external $\overline \lambda $ measures of arbitrary sets are determined by the measure $\lambda :\ r_{{\mathcal F},{\mathcal G}} \to [0,+\infty ) $ on the family of all subsets of the space $E$. Then ${\mathcal R}_{{\mathcal E},{\mathcal F}}=\{ A\subset E:\ {\underline \lambda }(A)={\overline \lambda }(A)\ in {\mathbb R}\}$.

\begin{theorem}\label{Th3.2.} \cite{GS22}
{\it
The measure $\lambda _{{\mathcal F},{\mathcal G}}:\ {\mathcal R}_{{\mathcal F},{\mathcal G}}\to [0,+\infty )$ is invariant relative to any symplectomorphism
$\Phi :\ E\to E$ such that for every
$k\in \mathbb N$ mapping ${\bf P}_{E_k}\Phi :\ E\to E_k$ does not depend on the value of the projection
${\bf P}_{E_k^{\bot }}x$ and the map ${\bf P}_{E_k}\Phi :\ E_k\to E_k$ is continuously differentiable on the space $E_k$.}
\end{theorem}

\begin{theorem}\label{r-space} For each $p\in [1,+\infty ]$
the space $L_p(E,{\mathcal R}_{{\mathcal E},{\mathcal F}},\lambda _{{\mathcal E},{\mathcal F}},{\mathbb C})$, defined as the completion in the $L_p$-norm of the space
$S(E,{\mathcal R}_{{\mathcal E},{\mathcal F}},\lambda _{{\mathcal E},{\mathcal F}},
{\mathbb C})$ of equivalence classes of simple functions is a non-separable Banach space. In this case, $L_p^*(E,{\mathcal R}_{{\mathcal E},{\mathcal F}},\lambda _{{\mathcal E},{\mathcal F}},{\mathbb C })=L_{q}(E,{\mathcal R}_{{\mathcal E},{\mathcal F}},\lambda _{{\mathcal E},{\mathcal F}},{\mathbb C})$ for $p\in [1,+\infty )$, where $q={p\over {p-1}}$.

\end{theorem}

\begin{proof}
Let $p\in [1,+\infty ]$. Let us define the space $L_p(E,{\mathcal R}_{{\mathcal E},{\mathcal F}},\lambda _{{\mathcal E},{\mathcal F}},{\mathbb C}) $ as follows (see \cite{DSh}). Along the ring ${\mathcal R}_{{\mathcal E},{\mathcal F}}$
let us define the linear space of simple functions as a linear space over the field $\mathbb C$ of linear combinations of indicator functions of disjoint sets from the ring
$$
S(E,{\mathcal R}_{{\mathcal E},{\mathcal F}},{\mathbb C})=
\{ \sum\limits_{k=1}^m\alpha _k\chi _{B_k},\ m\in {\mathbb N},\, \alpha _k\in {\mathbb C},\, B_k\ in {\mathcal R}_{{\mathcal E},{\mathcal F}},\, \forall \, k=1,...,m ,\ B_j\bigcap B_k=\emptyset \ {\rm if }\j\neq k\}.
$$
On the space of simple functions we define $S(E,{\mathcal R}_{{\mathcal E},\lambda _{{\mathcal E},{\mathcal F}},{\mathcal F}},{\mathbb C})$ the functional \begin{equation}\label{no5}
n_p(f)=\left(\int\limits_E|f|^pd\lambda _{{\mathcal E},{\mathcal F}}\right)^{1\over p},\ p\in [1 ,+\infty );\quad n_{\infty }(f)=\lim\limits_{p\to +\infty }\left(\int\limits_E|f|^pd\lambda _{{\mathcal E} ,{\mathcal F}}\right)^{1\over p},
\end{equation}
where $\int\limits_E|\sum\limits_{k=1}^m\alpha _k\chi _{B_k}|^pd\lambda _{{\mathcal E},{\mathcal F}}=\sum\ limits_{k=1}^m|\alpha _k|^p \lambda _{{\mathcal E},{\mathcal F}}({B_k})$ for each simple function of the form $f=\sum\limits_{ k=1}^m\alpha _k\chi _{B_k}$.
The functional $n_p :\ S(E,{\mathcal R}_{{\mathcal E},{\mathcal F}},{\mathbb C})\to \mathbb R$, as can be easily verified, is a seminorm on the space $S(E,{\mathcal R}_{{\mathcal E},{\mathcal F}},{\mathbb C})$.

Using the non-negative finitely additive measure $\lambda _{{\mathcal E},{\mathcal F}}$ we introduce on the space $S(E,{\mathcal R}_{{\mathcal E},{\mathcal F }},{\mathbb C})$ the equivalence relation $f\sim g\ \Leftrightarrow \ \lambda _{{\mathcal E},{\mathcal F}}(\{ x\in E:\ f(x) \neq g(x)\})=0$. Then the set $S_0(E,{\mathcal R}_{{\mathcal E},\lambda _{{\mathcal E},{\mathcal F}},{\mathcal F}},{\mathbb C}) =\{ f\in S(E,{\mathcal R}_{{\mathcal E},{\mathcal F}},{\mathbb C}):\ f\sim 0\}$ is a subspace in a linear space of simple functions. Let us introduce the space of equivalence classes of simple functions
$$
S(E,{\mathcal R}_{{\mathcal E},{\mathcal F}},\lambda _{{\mathcal E},{\mathcal F}},{\mathbb C})=S( E,{\mathcal R}_{{\mathcal E},{\mathcal F}},{\mathbb C})/ S_0(E,{\mathcal R}_{{\mathcal E},{\mathcal F }},\lambda _{{\mathcal E},{\mathcal F}},{\mathbb C}).
$$
If $f,g\in S(E,{\mathcal R}_{{\mathcal E},{\mathcal F}},{\mathbb C})$ and $f\sim g$, then $n _p (f)=n_p(g)$. Therefore the functional $n_p:\\ S(E,{\mathcal R}_{{\mathcal E},{\mathcal F}},\lambda _{{\mathcal E},{\mathcal F}},{\mathbb C})\to \mathbb R$ corresponding to each element $f\in S(E,{\mathcal R}_{{\mathcal E},{\mathcal F}},\lambda _{{\mathcal E},{\mathcal F}},{ \mathbb C})$. The value of the functional (\ref{no5}) on one of the representatives of the equivalence class $f$ is well defined and is a norm on the space $S(E,{\mathcal R}_{{\mathcal E}, {\mathcal F}},\lambda _{{\mathcal E},{\mathcal F}},{\mathbb C})$. Completion $L_p(E,{\mathcal R}_{{\mathcal E},{\mathcal F}},\lambda _{{\mathcal E},{\mathcal F}},{\mathbb C})$ the linear normed space $(S(E,{\mathcal R}_{{\mathcal E},{\mathcal F}},\lambda _{{\mathcal E},{\mathcal F}},{\mathbb C }), n_p)$ is a Banach space in which the space of equivalence classes of simple functions is dense everywhere.

For each $p\in [1,+\infty]$ the space $L_p(E,{\mathcal R}_{{\mathcal E},{\mathcal F}},\lambda _{{\mathcal E}, {\mathcal F}},{\mathbb C})$ contains a continuum system of elements (these are indicator functions of sets $\Pi _{\sigma}$ of the form (\ref{Pi}) in which $B_i=[\sigma (i ),\sigma (i)+1)\times [0,1)$. Here $\sigma :\ {\mathbb N}\to \{-1,0\}$), distance according to the $L_p$-norm between which there is at least one, which ensures the nonseparability of the Banach space $L_p(E,{\mathcal R}_{{\mathcal E},{\mathcal F}},\lambda _{{\mathcal E},{\mathcal F }},{\mathbb C})$.

Let $p\in [1,+\infty )$ and $q={p\over {p-1}}$, $f\in S(E,{\mathcal R}_{{\mathcal E}, {\mathcal F}},\lambda _{{\mathcal E},{\mathcal F}},{\mathbb C})$. Then $$n_p(f)=\sup\{ |\int\limits_{E}f(x){\overline {g(x)}}d\lambda _{{\mathcal E},{\mathcal F} }(x)| :\ {g\in S(E,{\mathcal R}_{{\mathcal E},{\mathcal F}},\lambda _{{\mathcal E},{\mathcal F}},{\mathbb C}),\ n_q(g)\leq 1} \}.$$
Since the space of equivalence classes of simple functions is dense in the spaces $L_p$ and $L_q$, then $L_p^*=L_q$ and for each function $f\in L_p(E,{\mathcal R}_{{\mathcal E}, {\mathcal F}},\lambda _{{\mathcal E},{\mathcal F}},{\mathbb C})$ the equality is true
$n_p(f)=\sup\{ |\int\limits_{E}f(x){\overline {g(x)}}d\lambda _{{\mathcal E},{\mathcal F}}(x)| :\ {g\in L_q(E,{\mathcal R}_{{\mathcal E},{\mathcal F}},\lambda _{{\mathcal E},{\mathcal F}},{\mathbb C}),\ n_q(g)\leq 1}\}.$
\end{proof}

{\bf Note}.
For $p=2$ the measure $\lambda _{{\mathcal E},{\mathcal F}}:\ {\mathcal R}_{{\mathcal E},{\mathcal F}}\to [0, +\infty )$ defines the Hilbert space ${\mathcal H}=L_2(E,{\mathcal R}_{{\mathcal E},{\mathcal F}},\lambda _{{\mathcal E},{ \mathcal F}},{\mathbb C})$ as a completion in the Euclidean norm $n_2$ of the  space
$S(E,{\mathcal R}_{{\mathcal E},{\mathcal F}},\lambda _{{\mathcal E},{\mathcal F}},{\mathbb C})$ of equivalence classes  of simple functions.

\section{Invariance of the symplectic measure with respect to Hamiltonian flows}
Let $h:E\to R$ be a non-degenerate quadratic Hamiltonian function on the Euclidean space $E$. The symmetric quadratic function on $E$ generated by a quadratic form $h$ has a canonical basis ${\mathcal E}={\mathcal F}\bigcup \mathcal G$, in which the quadratic form has a diagonal form. Let us also assume that the basis $\mathcal E$ is the canonical basis for the symplectic form $\bf J$ on the space $E$.
The flow ${\bf \Phi }_t,\, t\in \mathbb R$, given by the quadratic Hamiltonian $h$,
defines the one-parameter group
$$
{\bf U}_{{\bf \Phi }_t}u(x)=u({\bf \Phi }_{-t}(x)),\quad x\in E,\quad u\in S(E,{\mathcal R}_{\mathcal E},{\mathbb C}),\quad t\in {\mathbb R},
$$
This is linear isometries of the space of simple functions $S(E,{\mathcal R}_{\mathcal E},{\mathbb C})$ onto itself. The isometry group ${\bf U}_{{\bf \Phi }_t},\, t\in \mathbb R$ defined on a dense linear subspace $S(E,{\mathcal R}_{\mathcal E},{\mathbb C})$ in the space ${\mathcal H}_{\mathcal E}$, can be uniquely extended by continuity to a unitary group in the space ${\mathcal H}_{\mathcal E}$. Its acting according to the rule
\begin{equation*}
{\bf U}_{{\bf \Phi }_t }u(x)=u({\bf \Phi }_{-t}(x)),\quad t\in {\mathbb R},\ quad u\in {\mathcal H}_{{\mathcal F}, {\mathcal G}},\quad x\in E,
\end{equation*}
and called the Koopman representation of the Hamiltonian flow $\bf \Phi$.

Let us give examples of Hamiltonian flows that preserve the measure $\lambda _{\mathcal F,G}$.

{\bf Example 1.}
{\it A countable set of non-interacting two-dimensional Hamiltonian systems} is represented by the Hamiltonian function
\begin{equation}\label{BK}
\tilde {\mathbb H}(p,q)=\sum\limits_{k\in {\mathbb N}} \phi _k(p_k,q_k),\quad (p,q)\in E.
\end{equation}
Here $\{\phi _k\}$ is a sequence of functions
$\phi _k:\ E_k\to \mathbb R$, which are continuously differentiable for each
$k\in \mathbb N$ and satisfy the condition
\begin{equation*}
\sum\limits_{k=1}^{\infty }M_k< \infty , \end{equation*}
where $M_k=\sup\limits_{(p,q)\in {\mathbb R}^2}\left(|\phi _k(p,q)|^2
+|{{\partial }\over {\partial p_k}}\phi _k(p,q)|^2+|{{\partial }\over {\partial q_k}}\phi _k(p,q)| ^2\right)^{1\over 2}
$.

Under  assumptions made, the Hamilton function (\ref{BK}) is a continuously Frechet differentiable function on the space $E$. Consequently, the Hamiltonian function (\ref{BK}) generates a Hamiltonian flow $\Phi _{\tilde {\mathbb H}}$ in the space $E$. According to Theorem \ref{Th3.2.}, the flow $\Phi _{\tilde {\mathbb H}}$ preserves the measure $\lambda _{{\mathcal F},{\mathcal G}}$.

{\bf Example 2}.
If the Hamilton function ${\tilde {\mathbb H}}$ is a continuous linear functional on the space $E$, i.e. $ {\tilde {\mathbb H}}(z)=(h,{\bf J}z)_E$ for some $h\in E$, then
\begin{equation*}
\Phi _{\tilde {\mathbb H}}(t)z=z+th,\ z\in E,\, t\in \mathbb R.
\end{equation*}

{\bf Example 3}.
{\it Harmonic oscillator}.
Let
$\bf H$ is a linear self-adjoint operator in the space $H$ with a discrete spectrum
$\{a _k\}$ and ONB from eigenvectors
$\{ h_k\}={\mathcal H}$. Let ${\bf R}:\ H\to E$ be a reification of the space $H$. Then
\begin{equation*}
{\mathbb H}=\sum\limits_{k\in \mathbb N}a _k(p_k^2+q_k^2)=\sum\limits_{k\in \mathbb N}a _k|z_k|^2,
\quad (q,p)\in E_1=\{ (q,p)\in E:\ \sum\limits_{k=1}^{\infty }|a _k|(p_k^2+q_k^2) <+\infty\},
\end{equation*}
where $\ z=p\oplus q\in E$ and $z={\bf R}u$.
The Hamilton function
$\mathbb H$ generates a Hamiltonian flow $\bf \Phi$ in a symplectic space
$E={\bf R}(H)$. According to the Theorem \ref{Th3.2.}, the flow $\bf \Phi$ preserves the measure $\lambda _{{\mathcal F},{\mathcal G}}$. The flow $\bf \Phi$ is given by
\begin{equation*}
{\bf \Phi}_t(q,p)=(\cos ({\bf A} t)q-\sin ({\bf A} t)p, \sin ({\bf A} t)q+\ cos({\bf A} t)p),
\end{equation*}
where $t\in {\mathbb R},\ (q,p)\in E,$ $\bf A$ is a self-adjoint operator in the space $E$ such that
\begin{equation}\label{eugenf}{\bf A} f_j=a _jf_j,\ {\bf A} g_j=a _jg_j,\ j\in \mathbb N.\end{equation}

Let us introduce action-angle coordinates
$$q_k=\rho _k\cos \phi_k,\, p_k=\rho _k\sin \phi _k,\ k\in \mathbb N, \quad (\rho _k,\phi _k)\in (0,+ \infty )\times {\mathbb R}|_{{\rm mod}\, 2\pi},$$
The Hamiltonian flow $\Phi$ in the representation of this coordinates is given by a one-parameter family of mappings $\hat {\bf \Phi }_t,\ t\in \mathbb R$:
\begin{equation*}
\hat {\bf \Phi}_t(\rho ,\phi )=(\rho , \phi +a t),
\end{equation*}
where $ t\in {\mathbb R},\ (\rho,\phi)\in \ell _2^+\times ({\mathbb R}|_{{\rm mod}\, 2\pi})^ {\mathbb N}$ and $a =\{ a _k\}\in {\mathbb R}^{\mathbb N}$.

{\bf Example 4}. {\it Hyperbolic oscillator} is a
Hamiltonian system. The Hamiltonian function is densely defined on the phase space $(E,\omega )$ by the equality
\begin{equation}\label{hype}
\tilde {\mathbb H}={1\over 2}\sum\limits_{k\in {\mathbb N}} a _k(p_k^2-q_k^2),\qquad (q,p)\in E_1 =\{ (q,p)\in E:\ \sum\limits_{k=1}^{\infty }|a _k|(p_k^2+q_k^2)<+\infty \}.
\end{equation}
Here $a \equiv \{a _k\}\in \mathbb R^N$.

Let $z_0=(p_0,q_0)\in E$ be the initial point of the phase trajectory
\begin{equation}\label{Tr}
{\bf \Psi }_t(z_0)=z(t ,z_0),\ t\in (T_*,T^*).
\end{equation}
Then the trajectory has the form $z(t,z_0)=(q(t,z_0),p(t,z_0)),\ t\in (T_*,T^*)$, where
$$
p_k(t,z_0)=p_{0,k}{\rm ch} (a _kt)+q_{0,k}{\rm sh} (a _kt), \ k\in {\mathbb N},
$$
\begin{equation}\label{osc-tra}
q_k(t,z_0)=q_{0,k}{\rm ch} (a _kt)+p_{0,k}{\rm sh} (a _kt);\ t\in (T_*,T^* ).
\end{equation}

\begin{lemma}\label{Le3.5}\cite{GS22} {\it The interval $(T_*,T^*)$ of existence of a trajectory (\ref{Tr}) in the space $E$ is the real line $\mathbb R$ if and only if $a\in \ell_{\infty }$.

If $a\in \ell_{\infty }$, then the phase flow of the Hamiltonian system
(\ref{hype}) in a symplectic space
$(E,\omega)$ preserves the measure $\lambda _{{\mathcal F},{\mathcal G}}$ and is given by
\begin{equation}\label{hypeosc}
{\bf \Psi}_t(q,p)=({\rm ch} ({\bf A} t)q+{\rm sh} ({\bf A} t)p, {\rm sh} ({ \bf A} t)q+{\rm ch}({\bf A} t)p),\ t\in \mathbb R,
\end{equation}
where the self-adjoint operator $\bf A$ in the space $E$ is given by the equalities (\ref{eugenf}).}
\end{lemma}
{The statement of  lemma \ref{Le3.5} is a consequence of the equalities ({\ref{osc-tra}}) and the theorem \ref{Th3.2.}. The emergence of a singularity (going to infinity) in a finite time for a flow trajectory (\ref{Tr}) is described in the work \cite{GS22}.}

In contrast to the case of a hyperbolic oscillator with a limited set of frequencies. And it also contrast to the case of a harmonic oscillator. A densely defined Hamiltonian vector field on the space $E$ does not allow defining the group of Hamiltonian transformations of the space $E$ and the space $E_2=\{ (q,p)\ in E:\ \sum\limits_{k=1}^{\infty }|a _k|^2(p_k^2+q_k^2)<+\infty \}$.

{Recall, that the phenomenon of a gradient explosion of the solution of an evolutionary nonlinear partial differential equation
consists in the existence of a solution to the evolution equation on a limited time interval. The gradient of that is unlimited in the norm of the Banach space of values of the solution.}
A gradient explosion is observed when studying solutions to gas dynamics equations (Hopf equations),
when studying the phenomenon of self-focusing for solutions of the nonlinear Schrodinger equation \cite{Zhidkov, Bourgain}.
We give an example of a system of hyperbolic oscillators as a linear Hamiltonian system whose solutions admit the phenomenon of gradient explosion.

The Hamiltonian system of hyperbolic oscillators on the phase space $(E,\omega _{\bf J})$ can be considered in terms of a quantum system on the complexification $H$ of the space $E$. Its described by the wave function $u=q+ip$. With this approach, the energy functional is expressed through the wave function by the equality
\begin{align*}
h(q,p)=&-\frac{1}{4}\sum\limits_{k=1}^{\infty}[({\bf \Delta} u_k+{\bf \Delta} \bar u_k) (u_k+\bar u_k)+({\bf \Delta} u_k-{\bf \Delta} \bar u_k) (u_k-\bar u_k)]=
-{\rm Re}({\sqrt {\bf \Delta }} u,{\sqrt {\bf \Delta }} \bar u)_H.
\end{align*}
Here $\bf \Delta $ is a self-adjoint operator in the space $H$ with a simple discrete non-negative spectrum $\sigma (\bf \Delta )=\{ \omega _k\}$, and ${\sqrt {\bf \Delta }}$ is a non-negative square root of the operator $\bf \Delta $.

Let $\{ \psi _k\}$ be an orthonormal basis of the eigenvectors of the operator $\bf \Delta $. An arbitrary vector $u\in H$ admits the expansion $u=\sum\limits_{k=1}^{\infty }(q_k+ip_k)\psi _k=q+ip$. Then the Hamilton equations generated by the Hamilton function (\ref{hype}) take the form of the equation
$$
i\frac{d}{dt}u(t)={\bf \Delta }\bar u(t),\, t\in {\mathbb R},
$$
which is Hamiltonian, but is not the Schrodinger equation because, firstly, it is not linear over the field of complex numbers, and, secondly, it is not conservative.

The observed unlimited growth of the kinetic energy of a Hamiltonian system over a finite time is a phenomenon of gradient explosion (see \cite{Bourgain, Zhidkov}).
Phase trajectories of the Hamiltonian system (\ref{hype})
leave phase space in a finite time. For the Hamiltonian system under consideration (\ref{hype}), a natural symplectic extension into a locally convex space containing the space $E$ is found.

\section{Expansion of phase space and continuation of dynamics}

Let us extend the space $E=Q\oplus P\sim \ell _2\oplus \ell _2$ to the locally convex space of sequences ${\mathbb E}={\mathbb {R^N}}\oplus {\mathbb {R^N }}\supset E$. The locally convex space $\mathbb E$ is endowed with the Tikhonov topology, so the embedding $E\subset \mathbb E$ is dense and continuous. Let us extend the symplectic form $\omega$ from the space $E$ to the space $\mathbb E$
and flow $\bf \Psi$.
\subsection{Continuation of the symplectic form}

Into the  LCS $\mathbb E$ we can consider the space ${\mathbb R}^{\mathbb N}\oplus {\mathbb R}^{\mathbb N}$ of numerical sequences endowed with a metrizable topology of a pointwise (coordinate-wise) convergence. Then the embedding of the Hilbert space $E$ into the LCS $\mathbb E$ is continuous and dense. Then we can pose the question of extending the Hamiltonian flow from the space $E$ to the  space ${\mathbb E}$.

Let $E=Q\oplus P$, $({\mathcal E},{\mathcal F})$ be an ONB in ​​the Hilbert space $E$, in which the symplectic form $\omega _J$ has the canonical form:
$\omega _J ((\hat q,\hat p),(\tilde q,\tilde p))=\sum\limits_{k=1}^{\infty }(\hat q_k\tilde p_k-\tilde q_k\hat p_k)$.

We call the function $\Omega _{\bf J} :\ {\mathbb E}\times {\mathbb E}\supset D(\Omega _{\bf J})\to \mathbb R$
{\it pseudosymplectic form} on the space $\mathbb E$ if for each $z\in \mathbb {E}$ the set $D(z)=\{y\in \mathbb {E}:\ (z,y) \in D(\Omega _{\bf J})\}$ is a linear space, and the map $\Omega _{\bf J}(z,\cdot ):\ D(z)\to \mathbb R$ is a linear functional on the space $D(z)$ and the following conditions are satisfied:

1) if $y\in D(z)$, then $z\in D(y)$ and $\Omega _J(y,z)=-\Omega _J(z,y)$;

2) if $\Omega _J(z,y)=0\ \forall \ y\in D(z)$, then $z=0$;

3) if $z\in E$, then $D(z)\supset E$ and $\Omega _J(z,y)=\omega _J(z,y) \ \forall \ y\in E$.

The pair $({\mathbb E},\Omega)$, where $\mathbb E$ is a linear space containing $E$ and $\Omega $ is a pseudosymplectic form, will be called a pseudosymplectic space.

\begin{lemma}\label{lOmega}
The symplectic form $\omega _{\bf J}$ on the space $E$ can be extended to the pseudosymplectic form $\Omega _{\bf J}$ on the space $\mathbb E$.
\end{lemma}

\begin{proof}
For each $z=(q,p)\in \mathbb E$ we set $D_{\bf J}(z)=\{ (q',p')\in {\mathbb E}:\ \{ q_kp_k' -q_k'p_k\}\in l_1\}$. Then $D_{\bf J}(z)$ is a linear subspace in the locally convex space $\mathbb E$.

Let us define $\Omega _{\bf J}(z,\cdot ):\ D_{\bf J}(z)\to \mathbb R$ by the equality \begin{equation}\label{eqom}
\Omega _{\bf J}(z,z')=\sum\limits_{k=1}^{\infty}q_kp_k'-q_k'p_k.\end{equation}
Then the mapping $\Omega _{\bf J}$ defined on the set $D(\Omega _{\bf J})=\{ \bigcup\limits_{z\in \mathbb E}(z,D(z) )\}$ by the equality (\ref{eqom}), satisfies 1)-3) and is a pseudosymplectic form on the space $\mathbb E$. Consequently, the pseudosymplectic space $({\mathbb E},\Omega _{\bf J})$ is an extension of the symplectic space $(E,\omega _{\bf J}).$
\end{proof}

\subsection{Continuation of the flow $\bf \Psi $ to the space $\mathbb E$}

Let $\{ a_j\} \in \mathbb{R^N}$ be a sequence of parameters of the Hamilton function (\ref{hype}). Then the formula (\ref{hypeosc}) determines the flow
$\bf \Psi$ in the pseudosymplectic space $({\mathbb E},\Omega _{\bf J})$:
\begin{equation}\label{exflow}{\bf \Psi}_t(p,q)=(\ch ({\bf A }t)p+\sh ({\bf A }t)q,\, \ ch ({\bf A }t)q+\sh ({\bf A }t)p),\, t\in {\mathbb R}.\end{equation}

A one-parameter family of transformations (\ref{exflow}) is a continuation of the flow $\bf \Psi$ from the space ${ E}$ into the space $\mathbb E$ through times $T_*$ and $T^*$ of the existence of the E-valued solution.

It is easy to see that the extended flow $\bf \Psi$ in the space ${\mathbb E}$ preserves the pseudosymplectic form $\Omega _{\bf J}$.
Repeating the reasoning of point 3, it is easy to show that the measure $\lambda _{\mathcal F,G}$ on the space $E$ can be transformed into the measure ${\mathbf \lambda }_{\mathcal F,G}$ on the space $\mathbb E$ as a countable product of Lebesgue measures on two-dimensional subspaces $E_k,\ k\in \mathbb N$ of the space $\mathbb E$. Such a measure is invariant under the extended flow $\Psi$, which preserves the class of symplectic bars of the space $\mathbb E$ and the values of the measure on such symplectic bars.

\begin{theorem}\label{l6.1}\cite{GS22}
Let in the basis ${\mathcal E}={\mathcal F}\bigcup {\mathcal G}$ be
the ONB in a Hilbert space $E$ in which the symplectic form $\omega $ has the canonical form (\ref{simp}). Let ${\mathbb E}={\mathbb R}^{\mathbb N}\oplus {\mathbb R}^{\mathbb N}$ and let
$({\mathbb E},\Omega _{\bf J})$ is a pseudosymplectic extension of the symplectic space $(E,\omega _{\bf J})$. Let $E_0$ be a subspace of the space $\mathbb E$ whose vectors are linear combinations of vectors from the subspaces $E_k={\rm span}(e_k,f_k),\ k\in \mathbb N$.

Let the quadratic function $h$ be defined on a dense subspace $E_0$ in a locally convex space $\mathbb E$ by the equality $h=\sum\limits_{k=1}^{\infty }h_k$, where $h_k$ is a symmetric quadratic form on the two-dimensional subspace $E_k.$
Then the Hamiltonian vector field ${\bf v}={\bf J}\nabla h\, :\ E_0\to \mathbb E$ is densely defined on the subspace $E_0$.
The vector field $\bf v$ defines in the space $E_0$ a  Hamiltonian flow ${\bf \Phi }_t,\, t\in \mathbb R$, which admits a unique extension by continuity to a Hamiltonian flow $\bf \Phi$ on space $\mathbb E$. Moreover, the symplectic measure $\lambda _{\mathcal F,G}$ is invariant under the Hamiltonian flow ${\bf \Phi }_t,\, t\in \mathbb R$, on the space $\mathbb E$.
\end{theorem}

\begin{corollary}\label{c6.2}
{The one-parameter family ${\bf U}_{\bf \Phi }(t),\, t\in \mathbb{R}$, of linear operators in the space ${\mathcal H}_{{\mathcal E},{\mathcal F}}$ acting according to the rule
\begin{equation*}
{\bf U}_{\bf \Phi }(t)u(x)=u({\bf \Phi} (t)x),\qquad t\in {\mathbb R},\quad u\in {\mathcal H}_{{\mathcal E}, {\mathcal F}},\quad x\in E,
\end{equation*}
is the Koopman unitary representation of the Hamiltonian flow $\bf \Phi$
in the space ${\mathcal H}_{{\mathcal E},{\mathcal F}}$. }
\end{corollary}

To describe the subspaces of strong continuity of the Koopman unitary representation of the flow $\Psi$, we use not the measure ${\mathbf \lambda }_{\mathbb F,G}$, but an another invariant measure. To define a new invariant measure, we note that the phase flow $\Psi$ is a shift operator in action-angle coordinates $(r,\phi)\in ({\mathbb R}\times {\mathbb R})^{\mathbb N}$ related to the original coordinates $(q,p)\in \mathbb E$ using the replacement
${\bf G}:\, ({\mathbb R}\times {\mathbb R})^{\mathbb N}\to ({\mathbb R}\times {\mathbb R})^{\mathbb N }$, given by the equality $
(q,p)={\bf G}(\rho ,\phi ):\,
q=\rho {\rm ch} \phi,\, p=\rho {\rm sh} \phi $, i.e.
\begin{equation}\label{MGk}
(q_k,p_k)={\bf G}_k(r_k,\phi _k):\,
q_k=r_k{\rm ch} \phi _k,\, p_k=r_k{\rm sh} \phi _k,\quad k\in \mathbb N .
\end{equation}
According to (\ref{exflow}), in action-angle coordinates, the phase flow $\Psi $ is given by the mapping
\begin{equation}\label{actang}
\hat \Psi _t(r,\phi)=(r,\phi +at),\ t\in{\mathbb R}.
\end{equation}

\subsection {The measure on the space
$\mathbb E$ that is flow invariant}
Since we study translation-invariant finitely additive measures on the infinite-dimensional locally convex space of real-valued sequences, such measures will be introduced as infinite products of translation-invariant finitely additive measures on the real line.
On the real line $\mathbb R$ there is a unique, up to a scalar factor, non-negative translation invariant countably additive measure (Haar measure).
But in addition to the Lebesgue measure on the real line, finitely additive measures generated by Banach limits \cite{Usachev, Sukochev} have the property of translational invariance.

Using Banach limits, we define on the real line $\mathbb R$ non-negative normalized finitely additive measures that are invariant under translation. We called that measures as Banach measures.

Let $\beta $ be a Banach limit defined on the space $L_{\infty }({\mathbb R})$, i.e. $\beta \in L_{\infty }^*({\mathbb R})$ and this functional is non-negative, normalized by the condition $\beta ({\bf 1})=1$ and invariant under a shift, i.e. $\beta (\phi )=\beta ({\bf S}_h\phi )\ \forall \, \phi \in L_{\infty }({\mathbb R}),\, \forall \, h \in \mathbb R$, where ${\bf S}_h\phi (x)=\phi (x+h),\ x\in \mathbb R$.

If $\beta $ is a Banach limit defined on the space $L_{\infty }({\mathbb R})$, then $\nu _{\beta}$ is defined 
on the $\sigma$-algebra ${\mathcal L}(\mathbb R)$ of Lebesgue measurable sets of the real line by the equality
\begin{equation}
\label{mebe}
\nu _{\beta}(A)=\beta (\chi _A),\ A\in {\mathcal L}(\mathbb R),
\end{equation}
a shift-invariant non-negative normalized finitely additive Borel measure (whose domain contains the Borel $\sigma$-algebra) measure $\nu _{\beta}$ on the real line is defined.

\begin{definition}\label{def2}
A set $\Pi \subset \mathbb E$ is called an absolutely measurable hyperbolic block in the space $\mathbb E$ if
\begin{equation}\label{7}
\Pi =\{ (q,p)\in \mathbb E:
 (q_i,p_i)=(r_i\ch\phi_i,r_i\sh\phi_i),\, (r_i,\phi_i)\in A_i\times B_i,\, i\in {\mathbb N}\} ,
\end{equation}
where $A_i\times B_i\subset {\mathbb R}\times \mathbb R$, $A_i,B_i$ ; are Lebesgue measurable sets of the real line satisfying the following condition
\begin{equation*}
\sum\limits_{j=1}^{\infty} \ln _+(\nu _{2,\beta }(A_j\times B_j))<+\infty ,
\end{equation*}
in which $\nu _{2,\beta}(A_j\times B_j)=\nu _{\beta }(B_j)\int\limits_{A_j}|r|dr$.
\end{definition}

{\bf Note.} Since for each $k\in {\mathbb N}$ the Jacobian of the $k$-th change (\ref{MGk})
is equal to $|r_k|$, then the Lebesgue measure of the set $A_j\times B_j$ is equal to $\lambda _{2}(A_j\times B_j)=\lambda _{1 }(B_j)\int\limits_{A_j}|r |dr$. In Definition 9, we replace the Lebesgue measure $\lambda _1$ to the Banach measure $\nu _{\beta}$ in the  substitution of the the angular variable.

Let ${\mathcal K}_{\beta}(\mathbb {E})$ be the set of all measurable hyperbolic bars. Let the function ${\lambda}_{{\beta }} :\ {\mathcal K}_{\beta}(\mathbb{E})\to [0,+\infty )$ be defined by the equality
\begin{equation}\label{8}
{\lambda }_{{\beta }}(\Pi )=\prod\limits_{j=1}^{\infty }\nu _{2,\beta }(A_j\times B_j),\
\Pi \in {\mathcal K}_{\beta}(\mathbb{E}).
\end{equation}
\begin{lemma}\label{Lemma3.}
{\it The set function ${\lambda } _{\beta }:\ {\mathcal K}_{\beta}({\mathbb E})\to [0,+\infty )$ is additive and invariant with respect to flow (\ref{exflow}).}
\end{lemma}

\begin{proof}
The additivity of the set function ${\lambda}_{\beta } :\ {\mathcal K}_{\beta}({\mathbb E})\to [0,+\infty )$ was established in the work \cite{BS-direct}, Corollary 1. Its invariance with respect to the flow of a hyperbolic oscillator follows from the equality (\ref{actang}).
\end{proof}

Let $r_{\beta}$ be the ring generated by the family of sets ${\mathcal K}_{\beta}(\mathbb E)$.

\begin{theorem}\label{Theorem 4} Additive set function $\lambda _{\beta }:\ {\mathcal K}_{\beta}(\mathbb {E})\to [0,+\infty ) $ admits a unique additive extension to the ring $r_{\beta}$. Completion of the measure
$\lambda _{\beta }:\ r_{\beta } \to [0,+\infty )$ is the full measure of $\lambda _{\beta }:\ {\mathcal R}_{\beta }\to [0,+\infty )$, which is invariant under the Hamiltonian flow of the system of hyperbolic oscillators
(\ref{exflow}).
\end{theorem}

\begin{proof}
First, we prove the existence and uniqueness of the additive continuation of the additive function $\lambda _{\beta }:\ {\mathcal K}_{\beta}(\mathbb {E})\to [0,+\infty )$ from class ${ \mathcal K}_{\beta}(\mathbb {E})$ of measurable bars onto the ring $r_{\beta}$ generated by it.
The proof scheme is close to the construction in \cite{GS22}. First of all, note that the intersection of two sets from the class ${\mathcal K}_{\beta}(\mathbb {E})$ again belongs to this class. It follows that the collection of $\Lambda$ porous measurable hyperbolic bars, which are the difference of the block from ${\mathcal K}_{\beta}(\mathbb {E})$ and the finite union of such bars, is a semiring. Then the ring $r_{\beta}$ generated by a family of sets from the class ${\mathcal K}_{\beta}(\mathbb E)$ is generated by the semiring $\Lambda $ and each element of the ring $r_{\beta}$ represent it as a finite union of disjoint sets from the semiring $\Lambda$.

Let us denote by ${\Lambda }_j,\ j=1,2,\dots$ the collection of sets representable as the difference of a beam from the class ${\mathcal K}_{\beta}(\mathbb E)$ and the union from $ j$ bars from the same class (possibly having non-empty intersections). Through ${\mathbb V}_j$
let us denote sets representable as the union of $j$ bars from the class ${\mathcal K}_{\beta}(\mathbb E)$. Then you can use
induction method, by the additivity condition, extend the set function $\lambda _{\beta }:\ {\mathbb V}_1\to [0,+\infty )$ first to $\Lambda _1$, then from $\Lambda _1$ to $V_2$, then from $V_2$ to $\Lambda _2$, etc. As shown in \cite{GS22}, such a continuation does not depend on the choice of representing the set as a porous beam or a finite union of bars. Thus, the set function $\lambda _{\beta }:\ {\mathbb V}_1\to [0,+\infty )$ admits a unique additive extension $\lambda $ to the semiring $\Lambda$. The additive function $\lambda :\ \Lambda \to \mathbb R$ is non-negative by construction and, as is known, admits a unique additive continuation to a non-negative finitely additive measure $\lambda _{\beta }:\ r_{\beta }\to [0,+\infty )$.

The measure $\lambda _{\beta } :\ r_{\beta}\to [0,+\infty )$ defines the outer and inner measures on the set of all subsets of the space $E$ according to  formulas
$$
\overline \lambda _{\beta}(A)=\inf\limits_{B\in r_{\beta},\ B\supset A}\lambda _{\beta}(B),\quad
\underline \lambda _{\beta}(A)=\sup\limits_{B\in r_{\beta},\ B\subset A}\lambda _{\beta}(B),\quad A\subset E ,
$$
respectively. Then many
$$
{\cal R}_{\beta}=\{ A\subset E:\ \overline \lambda _{\beta}(A)=\underline \lambda _{\beta}(A)<+\infty \}
$$
is a ring. The ring ${\cal R}_{\beta}$ is called the completion of the ring ${r}_{\beta}$ with respect to the measure ${\lambda }_{\beta}$. Continuation of the measure $\lambda _{\beta } :\ r_{\beta}\to [0,+\infty )$ to the ring ${\cal R}_{\beta}$ by the formula $\lambda _{\beta }(A)= \overline \lambda _{\beta}(A)=\underline \lambda _{\beta}(A),\ A\in {\cal R}_{\beta}$, is called the completion of the measure $\lambda_{\beta }$.

Invariance property of the measure $\lambda _{\beta } :\ r_{\beta}\to [0,+\infty )$ with respect to the flow
(\ref{exflow}) follows from the definition of the set function (\ref{8}) due to the representation of the flow in the form (\ref{actang}). Because, firstly, the ring $r_{\beta}$ is invariant under the flow (\ref{actang}). And, secondly, the completion of the measure $\lambda _{\beta }:\ r_{\beta } \to [0,+\infty )$ inherits the invariance property by the following reason. Since the flow (\ref{actang}) preserves the measure  $\lambda _{\beta }:\ r_{\beta}\to [0,+\infty )$ it also preserves the approximation of the image of an arbitrary set from the inside and outside by sets from the ring $r_{\beta}$.
\end{proof}

Let ${\mathcal H}_{\beta }=L_2({\mathbb E},{\mathcal R}_{\beta },\lambda _{\beta },{\mathbb C})$ be the Hilbert space defined by the measure $\lambda _{\beta}$ according to the scheme outlined in the proof of Theorem 5. Let us study the unitary representation of the Hamiltonian flow
(\ref{hypeosc}) in the space ${\mathcal H}_{\beta }$.

\section{Koopman group of a hyperbolic oscillator}

\begin{lemma}\label{Lem5.1}\cite{S23} {\it Let $\{ a _k\} \in \ell _{\infty}$. Then the equality (\ref{hypeosc}) defines the flow $\Psi$ in the phase space $E$. The Koopman representation ${\bf U}_{\bf \Psi}$ in the space ${\mathcal H}_{{\mathcal F},{\mathcal G}}$ of the flow $\Psi $ of hyperbolic oscillators is a unitary group. The group ${\bf U}_{\bf \Psi}$ is strongly continuous if and only if the sequence $\{ a _k\}$ is finite.}
\end{lemma}

Let $\{ a _k\} \in \mathbb {R^N}$. Then the Hamiltonian flow in the extended phase space $\mathbb E$ is determined by the equality (\ref{exflow}).

Let us introduce  action-angle coordinates associated with  original phase coordinates of the space $\mathbb E$ using the mapping $${\bf G}:\, {\mathbb R}^{\mathbb N}\times {\mathbb R} ^{ \mathbb N}\to {\mathbb E}:\quad
q=r\ch \phi,\, p=r\sh \phi ,\quad r\in {\mathbb R} ^{\mathbb N},\, \phi \in {\mathbb R} ^{\mathbb N},$$ that is
$q_k=r_k\ch \phi_k,\, p_k=r_k\sh \phi _k,\ \forall \ k\in \mathbb N$.

Let $\lambda _{2,\beta }=\nu _{2,\beta }\circ {\bf G}_k^{-1}$ be the image of the measure $\nu _{2,\beta }$ when the mapping ${\bf G}_k:\ {\mathbb R}\times {\mathbb R}\to E_k$ from (\ref{MGk}).

Then if $\{ u_k\} :\ {\mathbb N}\to L_2(E_k,\lambda _{2,\beta},{\mathbb C})$, then
\begin{equation}\label{G_k}
\int\limits_{E_k}|u_k(q_k,p_k)|^2d\lambda _{2,\beta}(q_k,p_k)=\int\limits_{{\mathbb R}\times {\mathbb R}} |\tilde u_k(\rho_k,\phi _k)|^2d\nu _{2,\beta }(\rho _k,\phi _k) \ \ \forall \ k\in \mathbb N ,\end{equation}
where $\tilde u_k(\rho _k,\phi _k)=u_k({\bf G}_k(\rho _k,\phi _k)),\ (\rho _k,\phi _k)\in {\mathbb R }\times {\mathbb R}$.

For each $k\in \mathbb N$ the mapping $\bf G$ bijectively maps the set $({\mathbb R}\backslash \{ 0\})\times {\mathbb R}$ onto itself, but degenerates on the complement $ \{ 0\}\times {\mathbb R} $ of this set. 

The mapping $\bf G$ bijectively maps the set $({\mathbb R}\backslash \{ 0\})^{\mathbb N}\times {\mathbb R^N}$ onto itself, but degenerates on the complement of $\Sigma $ of this set. In this case, the projection $\Sigma _j$ of the degeneracy set $\Sigma$ onto the two-dimensional symplectic subspace $E_j={\rm span}(f_j,g_j)$ is the straight line $\Sigma _j=\{ (0,p_j),\ p_j \in {\mathbb R}\}$. Therefore, for each hyperbolic beam (\ref{7}) the set $\Pi \backslash \Sigma $ is a hyperbolic beam $\Pi '$ of the form (\ref{7}) with $A_j'=A_j\backslash \Sigma _j $ instead of $A_j$. Therefore, the value of the measure $\lambda _{\beta }$ for the intersection of any set $A\in {\cal R}_{\beta}$ with the set $\Sigma$ is defined and equal to zero, which will allow us to induce the measure $\lambda _ {\beta}\circ {\bf G}$ on the ring ${\bf G}^{-1}({\cal R}_{\beta })$ by means of the equality
$$(\lambda _{\beta}\circ {\bf G})(A)=\lambda _{\beta }({\bf G}(A))\ \forall \ A:\, {\bf G}(A)\in {\cal R}_{\beta }.$$

Then the measure $\lambda _{\beta}\circ {\bf G}:\ {\bf G}^{-1}({\mathcal R}_{\rho })\to {\mathbb R}_+ $ is non-negative and finitely additive.
The Hamiltonian system flow
(\ref{hype}) in action-angle coordinates is given by the equality (\ref{actang}), therefore the ring ${\bf G}^{-1}({\cal R}_{\beta })$ and the measure $\lambda _{\beta}\circ W$ are invariant under this flow.
Let us define a Hilbert space
$\tilde {\mathcal H}_{\beta }=
L_2(({\mathbb R}\times {\mathbb R})^{\mathbb N},{\bf G}^{-1}({\mathcal R}_{\beta }),\lambda _{ \beta}\circ {\bf G}, {\mathbb C}).$

According to the equality (\ref{G_k}), the mapping
$${\bf W_G}:\ \hat u\to u:\ u({\bf G}(\rho ,\phi ))=\hat u(\rho ,\phi ),\ (\rho ,\ phi )\in {\mathbb R}^{\mathbb N}\times {\mathbb R}^{\mathbb N},\ \hat u\in \tilde {\mathcal H}_{\beta},$$
is a unitary isomorphism of the Hilbert spaces ${\bf W_G}:\tilde {\mathcal H}_{\beta }\to {\mathcal H}_{\beta}$.

The Koopman representation of the Hamiltonian flow (\ref{actang}) in the space $\tilde {\mathcal H}_{\beta }$ has the form
$$
{\bf U}_{\hat {\bf \Psi }_t}\hat u(r,\phi)=\hat u(\hat{\bf \Psi }_t(r,\phi)), \ \ hat u\in \tilde {\mathcal H}_{\beta },\ t\in {\mathbb R}.
$$
In the set of Banach limits on the space $L_{\infty }({\mathbb R})$ we select a class of {\it Cesaro} Banach limits, representable as a limit over the ultrafilter of Cesaro averaging \cite{Sukochev}, i.e.
\begin{equation}\label{cheza}
\beta (f)=\lim\limits_{\digamma}\left({1\over {2x}}\int\limits_{-x}^xf(t)dt\right),\quad f\in L_{ \infty }(\mathbb R),
\end{equation}
where $\digamma$ is some ultrafilter focused at infinity, and $\lim\limits_{\digamma}$ is the limit on the ultrafilter $\digamma$.

Then if the Banach limit $\beta $ is given by the equality (\ref{cheza}), then for each function $f\in L_{\infty }({\mathbb R})$ due to the definition (\ref{mebe}) of the measure $\nu _{\beta}$ the equality is true
\begin{equation}\label{777}
\beta (f)=\lim\limits_{\digamma}\left({1\over {2x}}\int\limits_{-x}^xf(t)dt\right )=\int\limits_{\mathbb R}f(t)d\nu _{\beta}(t).
\end{equation}

\begin{lemma}\label{Lemma 5} {\it The Koopman representation of ${\bf U}_{\bf \hat \Psi}$ in the space $\tilde {\mathcal H}_{\beta}$ of the flow ( \ref{actang}) of the hyperbolic oscillator is a unitary group. Let, in addition, the Banach limit $\beta$ be Cesaro Banach limit. Then the unitary group ${\bf U}_{\bf \hat \Psi}${is continuous in the strong operator topology} if and only if the sequence of frequencies $\{ a _k\}$ in (\ref{hype}) is trivial. }
\end{lemma}

\begin{proof}
The unitarity of the group follows from the invariance of the measure $\lambda _{\beta}$ with respect to the flow (\ref{actang}). If we assume that $\lambda _k\neq 0$, then the group ${\bf U}_{\hat \bf \Psi}$ does not satisfy the condition of strong continuity. This follows from the existence of such a vector $\hat u
\in \tilde {\cal H}_{\beta}$ that the vector-valued function ${\bf U}_{\hat \bf \Psi}(t)u,\ t\in {\mathbb R},$ is not continuous at zero.
For the vector $\hat u\in \tilde {\cal H}_{\beta}$ we choose the function
$$\hat u(\rho ,\phi )=\left (\prod\limits_{j=1}^{\infty }\chi _{[a,b]}(r_j)\right ) e^{i \phi_k^2},\ \phi \in {\mathbb R}^{\mathbb N},\, \rho \in {\mathbb R}^{\mathbb N},$$
where $a,b\in \mathbb R$ are such that $\int\limits_a^b|r|dr=1$.
Then the vector-valued function ${\bf U}_{\hat \bf \Psi}(t)u,\ t\in {\mathbb R},$ of the real argument is not continuous, because the numerical function $g (t)=( {\bf U}_{\hat \bf \Psi}(t)u,u)_{\tilde {\cal H}_{\beta}},\ t\in \mathbb R$, has a removable discontinuity at $t=0$.

Indeed, $g(0)=\|\hat u\| ^2_{\tilde {\cal H}_{\beta }}=1$. While for each $t\neq 0$, by virtue of the definition of the measure $\lambda _{\beta }\circ {\bf G}$, the equality
$$g(t)=({\bf U}_{\hat \bf \Psi}(t)u,u)_{\tilde {\cal H}_{\beta}}=\int\limits_{ {\mathbb R}^{\mathbb N}\times {\mathbb R}^{\mathbb N}}\hat u(r,\phi +at){\overline {\hat u(r,\phi )} }d(\lambda _{\beta}\circ {\bf G}).$$
By virtue of equality (\ref{8}), which defines the measure $\lambda _{\beta}$, we obtain
$$g(t)=\int\limits_{\mathbb R}e^{i(\phi _k+a_kt)^2}e^{-i\phi _k^2}d\nu _{\beta }( \phi _k).$$ Since the Banach limit $\beta$ is Cesaro Banach limit, by virtue of (\ref{777}) we obtain that for all $t\neq 0$ the equality $g(t)=e^{ia_k^ holds 2t^2}\lim\limits_{\digamma}({1\over {2x}}\int\limits_{-x}^xe^{2ita_x\phi _k})d\nu _{\beta}\phi _k =0$, confirming the discontinuity at zero of the function $g$.
\end{proof}

Although the unitary group ${\bf U}_{\bf \hat \Psi}$ in the space $\tilde {\mathcal H}_{\beta}$ is discontinuous, we can determine its invariant subspaces of strong continuity by finding  eigenvalues and eigenvalues vectors of operators of this semigroup. To do this, we introduce the following space.

Let ${\mathcal K}^{rad}_{\beta}({\mathbb E})\subset {\mathcal K}_{\beta}({\mathbb E})$
is  a set of absolutely measurable hyperbolic bars that are radially symmetric in the sense that
in the representation (\ref{7}) $B_j=\mathbb R$ for all $j\in \mathbb N$.
Let $r_{\beta}^{rad}$ be the ring generated by a set of sets
${\mathcal K}^{rad}_{\beta}({\mathbb E})$ and ${\mathcal R}_{\beta}^{rad}$ -- completion of the ring $r_{\beta} ^{rad}$ in measure $\lambda _{{\beta}}$.
Let ${\mathcal H}^{rad}_{\beta }=L_2(E,{\mathcal R}^{rad}_{\beta },\lambda _{\beta },{\mathbb C}) $ and $\tilde {\mathcal H}^{rad}_{\beta }={\bf W_G}^{-1}({\mathcal H}^{rad}_{\beta })$.

The function $u(\phi _k)=e^{im\phi _k}\in L_2({\mathcal R},{\cal L}({\mathbb R}),\nu _{\beta },{\mathbb C})$ is the eigenfunction of the operator ${\bf L}_ku={{\partial u}\over {\partial \phi _k}}$ corresponding to the eigenvalue $im$ for each $m\in \mathbb R $. Since the phase flow $\hat {\bf \Psi}$ has the form (\ref{actang}), we have proven the following statement.
\begin{theorem}\label{Th7.5} {\it The unitary group ${\bf U}_{\bf \Psi }$ admits an invariant subspace ${\mathcal H}_{\Psi }$ such that the group ${\bf U}_{\bf \Psi}|_{{\mathcal H}_{\Psi }}$ is strongly continuous in the space ${\mathcal H}_{\Psi}$. The generator ${\bf H}_{\Psi }$ of a strongly continuous group ${\bf U}_{\bf \Psi}|_{{\mathcal H}_{\Psi }}$ has a continuum of eigenvalues
$$
a _{m_1,...,m_N}=m_1a _1+...+m_Na _N,\quad N\in {\mathbb N},\ m_1,...,m_N\in {\mathbb R}.$$
In this case, ${\mathcal H}_{\bf \Psi}\supset \oplus_{\vec m}{\mathcal H}_{\vec m}$, where $\vec m =\{ m_1,... ,m_N,0,0...\}$, $(\vec m,\phi )=m_1\phi _1+...+m_N\phi _N$ and
\begin{equation*}
{\mathcal H}_{\vec m}={\bf W_G}(\{ e^{i(\vec m,\phi)}g,\ g\in \tilde {\mathcal H}^{rad} _{\beta }\}) \subset {\rm Ker}({\bf H}_{\Psi }-a _{\vec m}{\bf I}).
\end{equation*}
}
\end{theorem}

\begin{proof}
It is easy to check that for each $\vec m =\{ m_1,...,m_N,0,0...\}$ the function $e^{i(\vec m,\phi)}g,\ g \in \tilde {\mathcal H}^{rad}_{\beta }$ is an eigenfunction of the operator ${\bf U}_{\hat {\bf \Psi}}(t)$ corresponding to the eigenvalue $e^ {it(m_1\lambda _1+...+m_N\lambda _N)}$. Therefore, for each $\vec m =\{ m_1,...,m_N,0,0...\}$ the subspace ${\cal H}_{\vec m}$
is invariant under the group $\bf U_{\Psi}$ and the restriction of the ${\bf U}_{\hat {\bf \Psi}}|_{{\cal H}_{\vec m}}$ group to this an invariant subspace is a strongly continuous group in the space ${\cal H}_{\vec m}$. The orthogonality of the subspaces ${\cal H}_{\vec m_1}$ and ${\cal H}_{\vec m_2}$ for $\vec m_1\neq \vec m_2$ follows from the direct calculation of the scalar product of functions of the form $e ^{i(\vec m_1,\phi)}g_1$ and $e^{i(\vec m_2,\phi)}g_2,\ g_1,g_2\in \tilde {\mathcal H}^{rad}_{ \beta }$. Consequently, the subspace $\oplus _{\vec m}{\cal H}_{\vec m}$ is invariant under the group $\bf U_{\Psi}$ and the restriction ${\bf U}_{\hat {\ bf \Psi}}|_{\oplus _{\vec m}{\cal H}_{\vec m}}$ is a strongly continuous group in the space $\oplus _{\vec m}{\cal H}_ {\vec m}$.
\end{proof}

\end{document}